\documentclass[conference,letterpaper]{IEEEtran}

\usepackage[utf8]{inputenc} 
\usepackage[T1]{fontenc}
\usepackage{url}
\usepackage{ifthen}
\usepackage{cite}
\usepackage[cmex10]{amsmath} 

\usepackage{amsfonts, amsmath, amssymb, amsthm}
\usepackage{array, makecell}
\usepackage{tikz, color, soul}
\usepackage{pgfplots}
\usetikzlibrary{plotmarks}
\usepackage{mathtools}
\usepackage[
font=small,labelfont=bf
]{caption}

\newtheorem{theorem}{Theorem}

\newtheorem{lemma}[theorem]{Lemma}

\theoremstyle{definition}

\theoremstyle{remark}

\theoremstyle{corollary}
\newtheorem{corollary}{Corollary}

\newcommand{\A}{{\mathcal A}}
\newcommand{\Ro}{{\mathcal R}}

\renewcommand{\epsilon}{\varepsilon}

\begin{document}
	\title{Duplication with transposition distance to the root for $q$-ary strings}
	\author{%
		\IEEEauthorblockN{\textbf{Nikita~Polyanskii}\IEEEauthorrefmark{1}, and \textbf{Ilya~Vorobyev}\IEEEauthorrefmark{2}\IEEEauthorrefmark{3}}
		\IEEEauthorblockA{\IEEEauthorrefmark{1}Institute for Communications Engineering, \\
			Technical University of Munich\\
			Munich, Germany 80333}	
		\IEEEauthorblockA{\IEEEauthorrefmark{2}Center for Computational and Data-Intensive Science and Engineering, \\
			Skolkovo Institute of Science and Technology\\
			Moscow, Russia 121205}
		\IEEEauthorblockA{\IEEEauthorrefmark{3}Advanced Combinatorics and Complex Networks Lab, \\
			Moscow Institute of Physics and Technology\\  Dolgoprudny, Russia 141701}
	%
		\IEEEauthorblockA{\textbf{Emails}: nikita.polyansky@gmail.com, vorobyev.i.v@yandex.ru}
	}
	
	\maketitle
	
	\begin{abstract}
	We study the duplication with transposition distance between strings of length $n$ over a $q$-ary alphabet and their roots. In other words, we investigate the number of duplication operations of the
	form $x = (abcd) \to y = (abcbd)$, where $x$ and $y$ are strings and
	$a$, $b$, $c$ and $d$ are their substrings, needed to get a $q$-ary string of length $n$ starting from the set of strings without duplications. For exact duplication, we prove that the maximal distance between a string of length at most $n$ and its root has the asymptotic order $n/\log n$. For approximate duplication, where a $\beta$-fraction of symbols may be duplicated incorrectly, we show that the maximal distance has a sharp transition from the order $n/\log n$ to $\log n$ at $\beta=(q-1)/q$. The motivation for this problem comes from genomics, where such duplications represent a special kind of mutation and the distance between a given biological sequence and its root is the smallest number of transposition mutations required to generate the sequence.
	\end{abstract}
	
		\section{Introduction}
 A genome sequence is the complete list of the nucleotides (A, C, G, and T for DNA genomes) that construct all the chromosomes of an individual or a species. Genomes are subjects to mutations. In this paper, we focus on a special type of mutation called duplication with transposition or  duplication mode IV according to classification~\cite{freeling2009bias} which create repeats by duplicating a substring
 and inserting the copy somewhere to the right from the original (e.g.,
 $A\underline{CG}T\to A\underline{CG}T\underline{CG}$). Such duplications are found in all eukaryotic genomes and constitute $25-40\%$ of most mammalian genomes.   Transposition-duplication differs from tandem repeat DNA that comes right after one another, they are dispersed throughout the genome and need not to be adjacent. The string that repeats can vary depending on the type of organism, and many other factors. 
 
 Any kind of repetitions in strings are fundamental objects in word combinatorics~\cite{lothaire1997combinatorics}, and a better understanding of these operations might give many applications in different fields, e.g., string matching algorithms, molecular biology, bioinformatics, and data compression~\cite{sayood2017introduction}. 
 
\subsection{Related Work}

A square in a string is a non-empty substring of the form $bb$. It is well-known that every binary string of length at least $4$ has a square (a tandem repeat), and a famous result by Thue~\cite{thue1906uber}
states that for alphabets of size at least $3$, there is an infinite
sequence without tandem repeats (of any length). This is evidently not true for repeats with transposition. 

Let us give a few additional notions to present an analogue of Thue's result for duplication with transposition. A repetition in a string $x$ is a pair of strings $(\hat{b},b)$ such that $\hat{b}b$ is a substring of $x$, $\hat{b}$ is non-empty, and $b$ is a prefix of $\hat{b}b$. The exponent of a repetition $(\hat{b},b)$ is $\frac{|b|+|\hat{b}|}{|\hat{b}|}$, where $|y|$ denotes the length of a string $y$. Tandem repeats are thus repetitions with exponent $2$. Dejean conjectured~\cite{dejean1972theoreme} that for a $q$-ary alphabet, there is an infinitely large $q$-ary sequence without repetitions with exponent larger than $r_q$, where $r_q:=q/(q-1)$ for $q=2$ and $q\ge5$, and $r_3:= 7/4$, $r_4:=7/5$. In other words,  such a string $y$ could not be represented in the form $y = (abcbd)$ if we require the property 
$$
\frac{2|b|+|c|}{|b|+|c|}>r_q \quad \Leftrightarrow \quad |c|< \frac{2-r_q}{r_q - 1}|b|,
$$
which means that $y$ could not be obtained from $x=(abcd)$ by a duplication with a relatively short transposition. Dejean's conjecture has successively been proved in~\cite{dejean1972theoreme,rao2011last, currie2011proof,carpi2007dejean}.

Tandem repeats have been studied extensively in recent years. One of the most
relevant papers to this work is the study of tandem duplication distance to the root for the binary strings~\cite{alon2017duplication}. It has been shown that to generate a binary
string of length $n$ starting from a square-free string from
the set  $\{0,1,01,10,101,010\}$, one needs a linear with $n$ number of tandem duplication operations. For the case of approximate duplication, where
a $\beta$-fraction of symbols may be duplicated incorrectly, the authors showed that the required number of duplications has a sharp transition from linear in
$n$ to logarithmic at $\beta= 1/2$. As the nature of duplications with transposition is a slightly different from one of tandem repeats, we could not apply ideas from~\cite{alon2017duplication} in a straightforward manner. However, after some modification of methods proposed in~\cite{alon2017duplication} we will show that some results can be transferred to our case.

In~\cite{lenz2019duplication, kovavcevic2018asymptotically}, the authors have proposed several constructions that correct tandem duplication errors of fixed length and derived a sphere-packing like bound. Error-correcting
codes for combating tandem duplication errors of short fixed length have been
introduced in~\cite{jain2017duplication}. The latter has been further investigated in~\cite{chee2018efficient, yehezkeally2018reconstruction, kovavcevic2019zero}.

	\subsection{Notation}
	Let us abbreviate a $q$-ary alphabet $\{0,1,\ldots, q-1\}$ by $\A_q$.	We denote $x=(x_1,\ldots, x_n)$ to be a string of length $|x|=n$ over $\A_q$. We say that $y$ is a substring of $x$ and $x$ has a distinguishable substring $y$ at position $i$ (for short, we write $x$ has $y^{(i)}$) if $y=(x_i,\ldots,x_j)$. A duplication of length $\ell$ at position $p$, $0\le p\le n-\ell$, with transposition $t$, $0\le t\le n-\ell - p$ in a string $x=(a b c d)$, with $|a|=p$, $|b|=\ell$, $|c|=t$, is defined by $f_{p,\ell,t}(x)=(abcbd)$. In what follows, we will consider all possible duplications of length at least $1$ with  admissible positions and transpositions. It is also helpful to define the inverse operation called deduplication with transposition, namely: the result of this operation applied to $y=(abcbd)$ is $x=(abcd)$, where $a$, $b$, $c$ and $d$ are substrings of $y$, and $b$ is non-empty.

Let $\A_q^{s}$ be the set of all strings over $\A_q$ of length $s$ and $\A_q^{\le n}$ be the union of $\A_q^{s}$ for $s$ that ranges from $1$ to $n$. Define the directed graph $G^{(n)}=(V,E)$ with the vertex set $V:=\A_q^{\leq n}$ such that there is an edge $(u,v)$ in $E$, $u,v\in V$, whenever $v$ is a result of duplication with transposition starting from $u$. The edge $(u,v)$ is considered to be directed from $u$ to $v$. We say vertex $u$ is a parent of $v$ if there is an edge $(u,v)\in E$. Vertex $v$ is called a child of $u$ if $u$ is a parent of $v$. We note that the set of vertices with indegree $0$ constitutes the set of roots abbreviated by 
$$
\Ro:=\bigcup\limits_{s=1}^q\{x\in\A_q^s: \, x_i\neq x_j  \text{ for any }i\neq j\}.
$$
Moreover, any vertex $v$ in graph $G^{(n)}$ can be reached by some path from the unique vertex from $\Ro$ which is called the root of $v$ and denoted by $root(v)$. Indeed, $root(v)$ coincides with the sequence of symbols from $\A_q$ appearing in $v$ for the first time when looking from left to right.

For any $v$, we define $f(v)$ to be the distance between $v$ and its root in graph $G^{(n)}$, i.e. the length of the shortest path in $G^{(n)}$ connecting $root(v)$ and $v$. The quantity of interest  is the maximum of distances over all strings of length at most $n$, that is
$$
f(n):=\max\limits_{v\in\A_q^{\le n}}f(v).
$$
To illustrate the values of $f(n)$ for the binary case in the range $1\le n \le 32$, we depict Table~\ref{tab:1}.
\begin{table}[t]
	\centering
	\caption{The maximum duplication distance to the root for the binary case and $1\le n \le 32$.}
	\label{tab:1}
	\begin{tabular}{||c||c|c|c|c|c|c|c|c|}
		\hline 
		$n$&  1&  2&  3&  4&  5&  6&  7&  8\\ 
		\hline 
		$f(n)$&  0&  1&  2&  2&  3&  4&  5& 5 \\ 
		\hline 
		\hline 
		$n$&  9&  10&  11&  12&  13&  14&  15& 16 \\ 
		\hline 
		$f(n)$& 6 & 6 & 7 & 7 & 8 & 8 & 8 & 9 \\ 
		\hline 
		\hline
		$n$&  17&  18&  19& 20 & 21 & 22 & 23 & 24 \\ 
		\hline 
		$f(n)$& 9 & 9 &  9&  10& 10 & 10 & 10 & 11 \\ 
		\hline
		\hline 
		$n$&  25& 26 & 27 & 28 & 29 & 30 & 31 & 32 \\ 
		\hline 
		$f(n)$& 11 & 11 & 11 & 12 & 12 & 12 & 12 & 13 \\ 
		\hline 
	\end{tabular} 	
\end{table}

The relative Hamming distance between two strings of equal length is the ratio between the number of positions at which the corresponding symbols are different and length of strings. Now let us define the  $\beta$-approximate duplication distance. To this end, suppose the directed graph $G_{\beta}^{(n)}=(V,E_{\beta})$ has the vertex set $V:=\A_q^{\leq n}$ and there is an edge $(u,v)$ in $E_{\beta}$, $u,v\in V$, whenever $v=(abc\hat{b}d)$, $u=(abcd)$ and the relative Hamming distance between $b$ and $\hat{b}$ is at most $\beta$. The length of the shortest path between a string $v$ and any
of its roots in $G^{(n)}_{\beta}$ is denoted by $f_{\beta}(v)$, and the maximum of $f_{\beta}(v)$
over all vertices from $\A_q^{\leq n}$ is denoted by $f_{\beta}(n)$. Clearly, we have $f(n)=f_{0}(n)$ and $f(n)\ge f_{\beta}(n)$ for $\beta>0$.
\subsection{Main Results}
Since for any string $v$, the length of a child of $v$ is at most twice as $|v|$, we have $f_{\beta}(n) \ge \log_2 (n/q)$. 
We strengthen this bound for $\beta = 0$ and present the main result of this paper in
\begin{theorem}\label{zeroDuplication}
	The function $f(n)=\Theta\left(\frac{n}{\log n}\right)$ as $n\to\infty$.  Moreover, we have
	$$
	\frac{n(1+o(1))}{2\log_q n} \le f(n) \le \frac{n(1+o(1))}{\log_q n}\quad \text{for }n\to\infty.
	$$
\end{theorem}
As for the case of approximate duplication, we show that there is a sharp transition of $f_{\beta}(n)$ at $\beta = (q-1)/q$.  
\begin{theorem}\label{betaDuplication}
	For any fixed $\beta\in(0,1)\setminus\{(q-1)/q\}$ and $n\to\infty$, the function $f_{\beta}(n)$ satisfies
	$$
	f_{\beta}(n)=\begin{cases}
	\Theta\left(\frac{n}{\log n}\right)\quad &\text{if }\beta<(q-1)/q,\\
	\Theta(\log n) \quad &\text{if }\beta>(q-1)/q.
	\end{cases}
	$$
\end{theorem}
\subsection{Outline} 
The remainder of the paper is organized as follows. The proof of Theorem~\ref{zeroDuplication} is given in Section~\ref{proofZeroDuplication}.  First, we give a greedy algorithm for finding a long duplicated substring in a string, what enables us to estimate the upper bound on $f(n)$. After that, we present a string that achieves $f(n)$ asymptotically up to a constant factor $1/2$.
We give the proof of upper and lower bounds for Theorem~\ref{betaDuplication} in Section~\ref{betaUpperDuplication} and Section~\ref{betaLowerDuplication}, respectively. Applying some results for error-correcting codes in the Hamming metric, the upper bound on $f_\beta(n)$ is strengthen with respect to $f(n)$. To obtain the lower bound on $f_\beta(n)$, we use some counting arguments. Finally, Section~\ref{conclusion} concludes the paper.  
\section{Proof of Theorem~\ref{zeroDuplication}}\label{proofZeroDuplication}
In this section we present the proof of lower and upper bounds on $f(n)$ in  Section~\ref{zeroUpperDuplication} and Section~\ref{zeroLowerDuplication}, respectively.
\subsection{Upper Bound on Duplication with Transposition Distance}\label{zeroUpperDuplication}
First of all, let us show an achievable length of duplicated substring in a string.
 \begin{lemma}\label{Mini invariant}
	Consider a $q$-ary string $x$ of length $n$ and a positive integer $k$ satisfying $$
	n-k+1> 2q^k + \frac{k^2q^{k/2}}{2}.
	$$
	Then $x$ can be represented as $x=(abcbd)$, where $|b|\ge k$.
\end{lemma}
\begin{proof}[Proof of Lemma~\ref{Mini invariant}] Recall that $x$ is said to have a distinguishable substring $y^{(i)}$ if $y$ is a substring of $x$ such that $y=(x_i,x_{i+1},\ldots,x_{i+|y|-1})$. Also we say that two distinguishable substrings $y^{(i)}$ and $z^{(j)}$ are intersecting if either $i<j$ and $i+|y|-1\ge j$ or $i\ge j$ and $i\le j +|z|-1$. Note that if there are  $>k$ distinct distinguishable substrings $\{b^{(i)},i\in I\}, I\subset [n],|I|>k,$ representing the same string $b$ of length $k$ in $x$, then we can find two of them which are non-intersecting, i.e. $x$ can be viewed as $x=(abcbd)$. Indeed, we can take the most left one distinguishable substring $b^{(i_1)}$ (for $i_1=\min\limits_{i\in I} i$) which can be intersected with at most $(k-1)$ other ones from $\{b^{(i)},i\in I\setminus\{i_1\}\}$.
	
Now we want to strengthen this property for certain strings. Define the period of a string $v=(v_1,\ldots, v_k)$ to be the smallest positive integer $p$ such that $v_i=v_{i+p}$ for all $i$, provided that $1\le i\le k-p$. The rational $k/p$ is called the exponent of $v$. One can easily see the number of strings in $\A_q^k$ with exponent at least $2$ is at most 
	$$
	\sum_{i=1}^{\lfloor k/2\rfloor}q^i\le \frac{k}{2} q^{k/2}.
	$$
	  Observe that if $x$ contains three distinguishable substrings $b^{(j_1)}, b^{(j_2)},b^{(j_3)}$  with $j_1<j_2<j_3$ and $j_3-j_1 < |b|$, then $b$ has exponent larger than $2$. Indead, inequality $j_3-j_1 < |b|$ implies that either $j_2-j_1<|b|/2$ or $j_3-j_2<|b|/2$. Without loss of generality, we assume that $j_2-j_1<|b|/2$ and 
	  $$
	  b_{1+t}=x_{j_1+t}=x_{j_2+t}=b_{j_2-j_1+1+t}
	  $$ for all $t$ provided that $0\le t< |b|-(j_2-j_1)$. Therefore, the exponent of $b$ is at least $|b|/(j_2-j_1)>2$. From this it follows that if $x$ contains at least $3$ distinguishable substrings representing the same $b$ with exponent at most $2$, then we have two of them which are non-intersecting. 
	  
	  The total number distinguishable substrings of length $k$ in $x$ is $n-k+1$. Therefore, the condition
	  	$$
	  	n-k+1> 2q^k + \frac{k^2q^{k/2}}{2}
	  	$$
	   implies that there exist  two distinguishable non-intersecting substrings representing the same string of length $k$ with exponent  either at most $2$, or larger than $2$. This completes the proof.
\end{proof}
For $k\ge14$, the inequality $n\ge 3q^k$ leads to the condition required in Lemma~\ref{Mini invariant}. Thus, Lemma~\ref{Mini invariant} guarantees that for $n\ge 3q^{14}$ and any $v$ in graph $G^{(n)}$ (in what follows, we assume that the length of $v$ is $n$), there is a parent $u$ of $v$ so that 
$$
|u|\le |v|-\log_q( |v|/3) < |v|-\log_q |v| + 2.
$$
 Therefore, the smallest distance between the string $v$ and some string $u$ of length $|u|\le|v|/q$ in graph $G^{(n)}$ is upper bounded by
$$
(|v|-|v|/q)\frac{1}{\log_q(|v|/q)-2}= \frac{(q-1)|v|}{q(\log_q |v|-3)}.
$$
We fix an arbitrary integer $k'>14$ and let $n\ge 3q^{k'}$. Applying iteratively the arguments above until the moment when the length of the parent is at most $3q^{k'}$, we obtain
$$
f(n)\le f(3q^{k'})+\sum_{i=1}^{\lfloor \log_q n \rfloor - k'+1} \frac{n(q-1)}{q^i (\log_q n-i-2)}.
$$
We compare two subsequent term $a_i$ and $a_{i+1}$ in the above sum and bound their ratio as follows
\begin{multline*}
\frac{a_i}{a_{i+1}}=\frac{q(\log_q n -i - 3)}{\log_q n -i - 2} \\ = q\left(1 - \frac{1}{\log_q n -i - 2}\right)
\ge q\frac{k'-4}{k'-3},
\end{multline*}
where we used the fact that $i\le \lfloor \log_q n \rfloor - k' + 1$. Therefore, 
$$
a_{i+1}\le \frac{(k'-3)a_i}{(k'-4) q}
$$
 and 
\begin{multline*}
f(n)\le 3q^{k'}+\sum_{i=1}^{\lfloor \log_q n \rfloor} \frac{n(q-1)(k'-4)^{i-1}}{q^i (\log_q n - 3) (k'-3)^{i-1}}\\
<3q^{k'}+ \frac{n(q-1)}{q (\log_q n - 3)(1-\frac{k'-3}{(k'-4)q})}.
\end{multline*}
Since the last inequality holds for any $k'\ge14$, we derive
\begin{corollary}
For $n\to\infty$, we have
$$
f(n) \le \frac{n}{\log_q n} (1+o(1)).
$$
\end{corollary}
\subsection{Lower Bound on Duplication with Transposition Distance}\label{zeroLowerDuplication}
Before we present a family of near-optimal strings, let us give an auxiliary statement which gives a connection between the distance from $y$ to the root and the number of distinct substring in $y$. A similar statement was proved in~\cite[Lemma 2]{alon2017duplication}.
\begin{lemma}\label{Semi invariant}
	Consider a $q$-ary string $y$ and a positive integer
	$k\ge q+1$, and let $N(y,k)$ denote the number of distinct substrings of length $k$ occurring in $y$. Then we have
	$$
	f(y)\ge \frac{N(y,k)}{2(k-1)}.
	$$
\end{lemma}
\begin{proof}[Proof of Lemma~\ref{Semi invariant}]
	We know that the result of deduplication applied to $y=(abcbd)$ is string $x=(abcd)$. Therefore, we have that $N(x,k)\ge N(y,k) - 2(k-1)$ since the only case in which a substring of length $k$ occurs in $y$, but not in $x$,  is when the substring includes symbols from either both $c$ and $b$ strings in $cb$ part, or both $b$ and $d$ strings in $bd$ part. The latter may happen at most $2(k-1)$ times. Therefore, the shortest path between $y$ and $root(y)$ should have length at least $$
	(N(y,k)-N(root(y),k))/(2(k-1)).
	$$
	Since $root(y)$ has length at most $q$, we have $N(root(y),k)= 0$.
	This completes the proof.
\end{proof}
Let us give a well-known definition of de Bruijn sequences~\cite{lothaire1997combinatorics}. A de Bruijn sequence of order $k$ on $\A_q$ is a cyclic sequence in which every possible string of length $k$ on $\A_q$ occurs exactly once as a substring ($k$ of them appear cyclically). It follows that the length of such a sequence is $q^k$. It is known that de Bruijn sequences exist for any order $k$ and alphabet size $q$.	We illustrate an example of de Bruijn sequence of order $2$ over $\A_3$ by sequence $001022112$. The number of distinct substrings of length $k$ for a de Bruijn sequence of order  $k$ is exactly $n-k+1$, where $n=q^k$. Therefore, this fact accompanied by Lemma~\ref{Semi invariant} leads to
\begin{corollary}
	For any de Bruijn sequence $x$ of order $k$, $k\ge q+1$, we have
	$$
	f(x)\ge \frac{q^k-k+1}{2(k-1)}.
	$$
	Therefore,
	$$
	f(n) \ge \frac{n}{2\log_q n} (1+o(1)) \text{ as }n\to\infty.
	$$
\end{corollary}
\section{Proof of Theorem~\ref{betaDuplication}}
We present the proof of lower and upper bounds on $f_{\beta}(n)$ in  Section~\ref{betaUpperDuplication} and Section~\ref{betaLowerDuplication}, respectively. The main idea here is to generalize the methods used in Section~\ref{proofZeroDuplication} and to apply some classic coding theory results.
\subsection{Upper Bound on Approximate Duplication Distance}\label{betaUpperDuplication}
First, we strengthen Lemma~\ref{Mini invariant} when $\beta> 0$ and $k$ is large. A similar method was used for tandem repeats in binary case in~\cite[Theorem 10]{alon2017duplication}.
\begin{lemma}\label{Mini beta invariant}
	Fix real $\beta>0$, and integers $n$ and $q$. Let $M(q,k,\beta)$ be the maximal cardinality of a $q$-ary code of length $k$ with relative Hamming distance $\beta$. Consider a $q$-ary string $x$ of length $n$ and a positive integer $k$ satisfying $n> kM(q,k,\beta)$. Then $x$ can be represented as $x=(abc\hat{b}d)$, where $|b|=|\hat{b}|\ge k$ and $d_H(b,\hat{b}) \le \beta |b|$.
\end{lemma}
\begin{proof}[Proof of Lemma~\ref{Mini beta invariant}] 
 First we split $x$ into $n':=\lfloor n /k\rfloor$ equal parts of length $k$, that is, $x=(v_1,v_2,\ldots,v_{n'},v')$ and $|v_{i}|=k$, $|v'|<k$. Let $C_x(q,k,\beta)$ be the largest collection of distinct substrings among $\{v_1,\ldots, v_{n'}\}$ with the minimum relative Hamming distance $\beta$. Its cardinality is at most $M(q,k,\beta)$. As $n' > M(q,k,\beta)$  there exist two strings $b$ and $\hat{b}$ from $\{v_1,\ldots, v_{n'}\}$ so that the relative Hamming distance between them is smaller than $\beta$.  In other words, $v$ can be represented as $v=(abc\hat{b}d)$, where $|b|=|\hat{b}|=k$.
\end{proof}
Define the maximal achievable rate and the asymptotic rate as follows
\begin{align*}
&R(q,k,\beta):= \frac{\log_q M(q,k,\beta)}{k},\\
&R(q,\beta):=\limsup_{k\to\infty} R(q,k,\beta).
\end{align*}
There are some upper bounds on $R(q,\beta)$ (e.g., see~\cite[Section 4.5]{roth2006introduction}). For example, for $\beta\in[0,\frac{q-1}{q})$, the Elias-Bassalygo bound states
$$
R(q,\beta)\le 1- H_q\left(\frac{q-1}{q}\left(1 - \sqrt{1-\frac{q\beta}{q-1}}\right)\right) ,
$$
where
$$
H_q\left(x\right): = -x\log_q(x)-(1-x)\log_q(1-x)+x\log_q(q-1).
$$
Let $\overline{r}$ be some upper bound on $R(q,\beta)$. We shall prove that  for $\beta\in[0,\frac{q-1}{q})$
$$
f_{\beta}(n)\le \frac{n \overline{r}}{\log_{q} n}(1+o(1))) \text{ as }n\to\infty.
$$
First, since $R(q,\beta)\le \overline{r}$, for any $\epsilon>0$, there exists sufficiently large $k_0=k_0(\epsilon)$ so that for $k\ge k_0$, $R(q,k,\beta) < \overline{r}+\epsilon$. Lemma~\ref{Mini beta invariant} says that for any $n\ge kq^{(\overline{r}+\epsilon)k}$ and any $v$ in graph $G_{\beta}^{(n)}$ (in what follows, we assume that the length of $v$ is $n$), there is a parent $u$ of $v$ such that 
$$
|u|\le |v| - \frac{1}{\overline{r}+\epsilon}\log_{q}|v|+\frac{1}{\overline{r}+\epsilon}\log_q\left(\frac{\log_q |v|}{\overline{r}+\epsilon}\right).
$$
Therefore, for $k\ge k_0$ and $n>kq^{(\overline{r}+\epsilon)k+1}$, the smallest distance between the string $v$ and some string $u$ of length $|u|\le|v|/q$ is upper bounded by
$$
(|v|-|v|/q)\frac{\overline{r}+\epsilon}{\log_q |v|-1- \log_q\left(\frac{\log_q |v|}{\overline{r}+\epsilon}\right)}.
$$
Thus, for any $k'> k_0$,  we have
\begin{multline*}
f_\beta(n)\le q^{(\overline{r}+ \epsilon)k'+1}\\
+\sum_{i=1}^{\lfloor \log_q n \rfloor - k' + 1}
\frac{n(q-1)(\overline{r}+ \epsilon)}{q^i\left(\log_q n - i - \log_q\left(\frac{\log_q |v|}{\overline{r}}\right)\right)}.
\end{multline*}
Finally, as the last inequality happens for any $\epsilon>0$ and $n>n(\epsilon)$, we conclude with 
\begin{corollary}
For any $\beta\in[0,\frac{q-1}{q})$, we have
$$
f_\beta(n)\le \frac{n\overline{r}}{\log_q n}(1+o(1)) \text{ as } n\to\infty,
$$	
where $\overline{r}$ is an upper bound on $R(q,\beta)$.
\end{corollary}
 
Now we recall~\cite[Section 4]{roth2006introduction} the Plotkin-type bound for $\beta > \frac{q-1}{q}$
$$
M(q,k,\beta) < \frac{\beta q}{\beta q - (q-1)}.
$$
Define $c$ to be $\lceil\beta q / (\beta q - (q-1))\rceil$. We divide a $q$-ary string $v$ from $G_\beta^{(n)}$ (in what follows, we assume that the length of $v$ is $n$) into $c$ parts of the same length $k:=\lfloor n/c\rfloor$ and the remaining part, i.e., $v$ can be represented in the way $v=(v_1,\ldots, v_c, v')$ and $|v_i|=k$, $|v'|< k$. From the Plotkin-type bound, there are two strings among $\{v_1,\ldots, v_c\}$ so that the relative Hamming distance between them is at most  $\beta$. In other words, there is a parent $u$ of $v$ with 
$$
|u|\le |v|- \frac{1}{c+1}|v|.
$$
Let $q'$ be equal to $(c+1)/c$. Therefore,  we have
$$
f_\beta(n)\le (c+1) +  \sum_{i=1}^{\lfloor \log_{q'} n \rfloor} 1\le c+1 + \log_{q'} n.
$$
We conclude with 
\begin{corollary}
	For any $\frac{q-1}{q}<\beta\le 1$, we have
	$$
	f_\beta(n)\le \log_{q'} n (1+o(1))\text{ as }n\to\infty,
	$$
	where $q':= (c+1)/c $ and $c:=\lceil\beta q / (\beta q - (q-1))\rceil$. 
\end{corollary}
\subsection{Lower Bound on Approximate Duplication Distance}\label{betaLowerDuplication}
The first attempt to generalize the construction from Section~\ref{zeroLowerDuplication} for the case of approximate duplication leads to the concept of error-correcting sequences~\cite{hagita2008error}. However, it appears quite difficult to prove an analogue of Lemma~\ref{Semi invariant}. In what follows, we will not provide a specific family of strings which achieve $\Theta(n/ \log n)$ distance to the root. In contrast, we prove that such a distance is $\Omega (n/ \log n)$ on average. A similar technique was used in~\cite{alon2017duplication}.

Recall that a $\beta$-approximate duplication of length $\ell$ at position $p$, $0\le p\le n-\ell$, with transposition $t$, $0\le t\le n-\ell - p$ in a string $x=(a b c d)$, with $|a|=p$, $|b|=\ell$, $|c|=t$, is defined by $f_{p,\ell,t,j}(x)=(abc\hat{b}d)$, where $\hat{b}$ is the $j$th word in the Hamming ball with center in $b$ and radius $\beta|b|$ (we assume the lexicographical order on $q$-ary words belonging to the Hamming ball). Now we observe that any path from $\Ro$ to a string from $\A_q^n$ in graph $G_{\beta}^{(n)}$ can be described by a sequence of quadruples $(p, \ell, t, j)$. Moreover, there are some constraints on quadruple. For example, if 
$$
x=(a b c d)\to y:=f_{p,\ell,t,j}(x) = (abc\hat{b}d),
$$
then 
$$
0\le p\le n-\ell, \quad 0\le t\le n-\ell - p,
$$
$$
\ell = |y|-|x|, \quad 0\le j < \sum\limits_{s=0}^{\lfloor \beta |b|\rfloor } {|b| \choose s}.
$$
Suppose that $f_{\beta}(n)\le f_\beta$ for $\beta\in [0, (q-1)/q)$, where $f_{\beta}$ is some positive integer. Since any string from $\A_q^n$ can be described (not uniquely) by a sequence of quadruples, satisfying the above constraints and having some length $f'$ not greater than $f_\beta$, and a root from $\Ro$, we shall count and upper bound the total number of such assignments by
$$
q!q\sum_{i=1}^{f_\beta} n^{2i} {n \choose i}q^{H_q(\beta)n}. 
$$
Here, we use the facts that the number of choices for $p$ and $t$ is at most $n$, the number of roots $|\Ro|\le q! q$,  the sum of $\ell$'s in the sequence of quadruples is at most $n$ and for $\beta\in[0,(q-1)/q)$, the size of the Hamming ball  with radius $\beta n$ can be  bounded by $q^{H_q(\beta)n}$ (e.g., see~\cite[Section 4]{roth2006introduction}). On the other hand, the above estimate should be greater than $q^n$. This implies the inequality
$$
q!qf_\beta n^{3f_\beta}q^{H_q(\beta)n} \ge q^n.
$$
Therefore, we come to
\begin{corollary}
	For $0<\beta <(q-1)/q$, the function
$$
f_\beta(n) \ge \frac{n(1-H_q(\beta))}{3\log_q n}(1+o(1)) \text{ as } n\to\infty,
$$
where $H_q(x)$ is the $q$-ary entropy function.
\end{corollary}
For $\beta > 1 - 1/q$, let us recall the arguments. Since for any $v$, the length of a parent of $v$ is at least half of $|v|$, we have $f_{\beta}(n) \ge \log_2 (n/q)$.  This completes the proof of Theorem~\ref{betaDuplication}.
\section{Conclusion}\label{conclusion}
In this paper, we study the duplication with transposition operation which may occur in a genome sequence. The basic question of our research is related to the maximal distance in graph $G^{(n)}$ between the set of roots $\Ro$ and $q$-ary strings of length at most $n$.  

We showed in Theorem~\ref{zeroDuplication} that the maximal distance, denoted by $f(n)$, satisfies the inequality $0.5 n/ \log_q n  \lesssim f(n)\lesssim n/ \log_q n$. Also, we proved in Theorem~\ref{betaDuplication} that for the case of $\beta$-approximate duplication, the maximal distance $f_\beta(n)$ has a sharp transition from $n/\log n$ order to $\log n$ at $\beta = (q-1)/q$. The behaviour of $f_{(q-1)/q}(n)$ for large values of $n$ remains an open question.
\section*{Acknowledgment}
We thank Yiwei Zhang for the introduction to the problem of the duplication distance to the root for the binary strings and Gregory Kucherov for the fruitful discussion on Dejean's conjecture. 

I. Vorobyev was supported in part by the Russian Foundation for Basic Research through grant no.~\mbox{20-01-00559}. N. Polyanskii was supported in part by the European Research Council under the EU’s Horizon 2020 research and innovation programme (grant No. 801434).
\bibliographystyle{IEEEtran}
\bibliography{duplication}
\end{document}